\documentclass[floatfix,twocolumn,notitlepage,amsfonts,amsmath,amssymb,nofootinbib,superscriptaddress]{revtex4-2}

\pdfoutput=1

\usepackage[utf8]{inputenc}
\usepackage[T1]{fontenc}
\usepackage{accents}
\usepackage[colorlinks]{hyperref}

\usepackage{booktabs}
\usepackage{longtable}

\usepackage[normalem]{ulem}
\usepackage{amsthm}
\usepackage{mathtools}
\usepackage{bbm}
\usepackage{bm}
\usepackage{verbatim}
\usepackage{graphicx}
\usepackage{tikz}

\usepackage{soul}
\usepackage{cancel}

\usepackage{float}

\usetikzlibrary{arrows}

\usepackage{lineno}

\newtheorem{theorem}{Theorem}
\newtheorem{proposition}[theorem]{Proposition}

\newtheorem{corollary}[theorem]{Corollary}

\newtheorem{lemma}[theorem]{Lemma}

\newtheorem{definition}[theorem]{Definition}

\def\bra#1{\mathinner{\langle{#1}|}}
\def\ket#1{\mathinner{|{#1}\rangle}}

\def\ketbra#1#2{\mathinner{|{#1}\rangle\!\langle{#2}|}}

{\catcode`\|=\active 
  \gdef\Braket#1{\left<\mathcode`\|"8000\let|\BraVert {#1}\right>}}
\def\BraVert{\egroup\,\mid@vertical\,\bgroup}

\DeclareMathOperator{\Tr}{Tr}

\DeclareMathOperator{\Span}{span}

\renewcommand\L{\mathcal{L}}

\newcommand{\cH}{\mathcal H}
\newcommand{\N}{\mathcal{N}}

\newcommand{\K}{\mathcal{K}}

\newcommand{\X}{\mathcal{X}}

\newcommand{\id}{\mathbbm{1}}

\newcommand{\Nset}{\mathcal N}
\newcommand{\supp}{\operatorname{supp}}
\newcommand{\range}{\operatorname{range}}

\newcommand*\patchAmsMathEnvironmentForLineno[1]{%
  \expandafter\let\csname old#1\expandafter\endcsname\csname #1\endcsname
  \expandafter\let\csname oldend#1\expandafter\endcsname\csname end#1\endcsname
  \renewenvironment{#1}%
     {\linenomath\csname old#1\endcsname}%
     {\csname oldend#1\endcsname\endlinenomath}}% 
\newcommand*\patchBothAmsMathEnvironmentsForLineno[1]{%
  \patchAmsMathEnvironmentForLineno{#1}%
  \patchAmsMathEnvironmentForLineno{#1*}}%
\AtBeginDocument{%
\patchBothAmsMathEnvironmentsForLineno{equation}%
\patchBothAmsMathEnvironmentsForLineno{align}%
\patchBothAmsMathEnvironmentsForLineno{flalign}%
\patchBothAmsMathEnvironmentsForLineno{alignat}%
\patchBothAmsMathEnvironmentsForLineno{gather}%
\patchBothAmsMathEnvironmentsForLineno{multline}%
}

\begin{document}

\title{All causally separable quantum processes \\ are quantum circuits with classical control of causal order}

\author{Julian Wechs}
\affiliation{Centre for Quantum Information and Communication (QuIC), \'Ecole Polytechnique de Bruxelles, C.P.\ 165, Universit\'e libre de Bruxelles, 1050 Bruxelles, Belgium}

\author{Alastair A.\ Abbott}
\affiliation{Univ.\ Grenoble Alpes, Inria, 38000 Grenoble, France}

\author{Cyril Branciard}
\affiliation{Univ.\ Grenoble Alpes, CNRS, Grenoble INP, Institut N\'eel, 38000 Grenoble, France}

\date{September 17, 2026}

\begin{abstract}
The concept of causal (non)separability describes whether the causal order between parties that perform local quantum operations is well-defined or indefinite. Causal (non)separability in the general multipartite setting was introduced and studied in Refs.~\cite{oreshkov16,Wechs2019}. We resolve an open problem from these earlier works by showing -- using a novel ``coherent teleportation technique'' -- that a sufficient condition for causal separability identified in Ref.~\cite{Wechs2019} is also necessary, and thus provides a complete characterisation of multipartite causal separability. A consequence of this result is that all causally separable processes admit a realisation as generalised quantum circuits in which the order between the operations is classically controlled, known as ``quantum circuits with classical control of causal order''~\cite{wechs18}.
\end{abstract}

\maketitle
\section{Introduction}
The investigation of quantum causal relations has emerged as an active area of research in quantum foundations and quantum information. In particular, it has been found that one can conceive of processes in which the causal order between quantum operations is no longer well-defined (see, e.g., Refs.~\cite{oreshkov12,araujo15,oreshkov16,branciard16,abbott16,Wechs2019,baumeler14a,baumeler16}, and Ref.~\cite{Costa26} for a recent review of the topic). 
Such processes are often formally described and studied in the \emph{process matrix framework}, an abstract, ``top-down'' approach, in which one considers multiple parties performing local quantum operations once and only once, without assuming an \textit{a priori} global causal order between them. The object that connects the parties and encodes their causal relations can be represented by a so-called \emph{process matrix}.  

The formal concept that distinguishes process matrices that are compatible with a well-defined causal order from those that are not is \emph{causal (non)separability}. This notion was initially introduced for the bipartite case in Ref.~\cite{oreshkov12}, and generalised to the multipartite case in~\cite{oreshkov16,Wechs2019}. In particular, the multipartite definition of causal separability accounts for the fact that the causal order can be definite, but \emph{dynamical}; that is, the order may not be fixed from the start, but may only be established ``on the fly'', as the parties act one after the other.

In Ref.~\cite{Wechs2019}, a decomposition of multipartite process matrices was proposed and shown to be sufficient for causal separability, while its necessity (beyond the bipartite and tripartite cases) was left as an open question. In this work, we answer this question in the affirmative and thus complete the characterisation of causal separability in the general multipartite case.

This result is of particular interest with regard to the question of which scenarios described in the abstract process matrix framework admit a physical interpretation or realisation. 
While the physicality of processes with indefinite causal order is a major research question and a topic of intense debate (see, e.g., Refs.~\cite{araujo17, Wechs21,Oreshkov18,Wechs23,Wechs25,Kabel24,maclean17,Vilasini24a,Vilasini24b,ormrod22,Paunkovic20})), the situation was not fully clear even for causally separable processes. In Ref.~\cite{Wechs21}, the ``top-down'' perspective of the process matrix framework was complemented by a constructive, ``bottom-up'' approach, and several classes of process matrices realisable by generalised quantum circuits were developed. In particular, \emph{quantum circuits with classical control of causal order} (QC-CCs) were introduced. These are generalised quantum circuits in which the order between operations is established ``on the fly'', in a classically controlled manner, as the circuit evolves. The class of process matrices realisable by QC-CCs was shown to coincide with the class satisfying the sufficient condition for causal separability from Ref.~\cite{Wechs2019}.

Our result therefore shows that causally separable process matrices and QC-CCs are equivalent, i.e., that all processes compatible with a well-defined causal order in the process matrix framework admit such an operational realisation -- as initially conjectured in Ref.~\cite{oreshkov16}, before the QC-CC class was even formalised and characterised.

The main technical ingredient that allows us to complete the proof is a ``coherent teleportation technique''. The multipartite definition of causal separability is recursive: intuitively, an $N$-partite process matrix is causally separable if, with some probability, one can identify a party that acts first, and for any operation performed by this initial party, the resulting $(N-1)$-partite process matrix is again causally separable. In Ref.~\cite{Wechs2019}, necessary conditions for causal separability were derived by choosing the operations of each initial party so as to ``teleport'' its systems to one of the remaining parties, but these conditions could not be shown to coincide with the sufficient condition. By instead choosing the operation of the initial party so as to teleport its systems ``coherently'' to all remaining parties, we are able to complete the proof.

We first give a concise overview of the relevant concepts, definitions and notation. For a more detailed discussion, we refer the reader to the literature, in particular, Ref.~\cite{Wechs2019}. We then present the proof of the main theorem and discuss its implications. 

\section{Background}

In this section, we recall the multipartite process matrix framework and the definition of causal separability, so as to make the paper self-contained. 
We adopt essentially the same notations and conventions as in Ref.~\cite{Wechs2019}.

\subsection{The process matrix framework}

In the following, we describe a generic quantum system by the space $X$ of Hermitian operators on a Hilbert space $\cH^X$. We only consider Hilbert spaces $\cH^X$ of finite dimension, denoted $d_X$. For composite systems, we often use concatenation, so that $XY \coloneqq X \otimes Y$ and $\cH^{XY} \coloneqq \cH^X \otimes \cH^Y$. The identity operator on $\cH^X$ is denoted by $\id^X$ (in general, superscripts of operators or vectors denote the spaces they belong to, which may be omitted when clear from the context). The notation $\id^{X\to X'} := \sum_i \ket{i}^{X'}\bra{i}^X$ is used to denote the ``identity operator'' between two isomorphic spaces $\cH^X$ and $\cH^{X'}$ whose computational bases $\{\ket{i}^X\}_i$ and $\{\ket{i}^{X'}\}_i$, respectively, are in one-to-one correspondence: $\id^{X\to X'}$ just ``reattributes'' states in $\cH^X$ to $\cH^{X'}$. We also denote by $\ket{\Phi^+}^{X/X'}:=\frac{1}{\sqrt{d_{X^{(\prime)}}}}\sum_i \ket{i}^X\otimes\ket{i}^{X'}$ the maximally entangled state across these spaces.

The general multipartite scenario we consider involves $N$ parties, denoted by $A_k$ for $k \in \N \coloneqq \{ 1, \ldots, N\}$. Each party receives an incoming physical system, performs a local operation, which may depend on a classical input $x_k$ and generates a classical output $a_k$, and sends away an outgoing physical system. The central object of interest is the multipartite correlation $P(\vec a | \vec x)$ established after many runs of the experiment (where $\vec x \coloneqq (x_1, \ldots, x_N)$ and $\vec a \coloneqq (a_1, \ldots, a_N)$). 

In the process matrix framework, the incoming and outgoing systems, as well as the local operations, are described by standard quantum theory. We denote the incoming and outgoing spaces of each party by $A_I^k$ and $A_O^k$, respectively. 
We use the shorthand notations $A_{IO}^k \coloneqq A_I^k \otimes A_O^k$, as well as, for a subset $\K \subseteq \N$ of parties, $A_{IO}^\K \coloneqq \bigotimes_{k \in \K} A_{IO}^k$ ($= \mathbb{R}$ if $\K = \emptyset$), $\id^\K \coloneqq \bigotimes_{k \in \K} \id^{A_{IO}^k} = \id^{A_{IO}^\K}$, and $\Tr_\K$ for the trace over all (incoming and outgoing) systems of the parties in $\K$, with $\Tr_{\emptyset}$ the identity operation and $\Tr_{\N}$ the full trace.
For notational simplicity, we identify the parties' names with their labels, and singletons of parties (e.g., $\{A_k\}$) with the parties themselves (e.g., $A_k$) or the corresponding label, so that $\N = \{ 1, \ldots, N\} \equiv \{ A_1, \ldots, A_N\}$, $\N \backslash \{A_k\} \equiv \N \backslash k$, $\Tr_{\{A_k\}} \equiv \Tr_k$, etc.

Each party's local operation is most generally described by a \emph{quantum instrument}~\cite{davies70}, that is, a collection of completely positive (CP) linear maps from $A_I^k$ to $A_O^k$, each of which is associated to one of the classical outcomes $a_k$, and such that the sum over the classical outcomes results in a completely positive trace-preserving (CPTP) map. Using the Choi-Jamio\l{}kowski (CJ) isomorphism~\cite{jamiolkowski72,choi75}, these CP maps can be represented as positive semidefinite matrices $M_{a_k|x_k}^{A_{IO}^k}$, with the corresponding CPTP maps $M_{x_{k}}^{A^{k}_{IO}}\coloneqq \sum_{a_{k}}M_{a_{k}|x_{k}}^{A^{k}_{IO}}$, which satisfy $\Tr_{A_O^k} M_{x_{k}}^{A^{k}_{IO}} = \id^{A_I^k}$. 

Rather than assuming an \emph{a priori} causal order between the parties, in the process matrix framework one only assumes the local validity of quantum theory and requires that the correlations $P(\vec a | \vec x)$ define a valid conditional probability distribution~\cite{oreshkov12}. Under these assumptions, along with multi-linearity in the CP-maps $M_{a_k|x_k}^{A_{IO}^k}$, the correlations can be shown to take the general form
\begin{equation} \label{eq:born_rule_Npartite}
P(\vec a | \vec x) = \Tr\left[ M_{a_1|x_1}^{A_{IO}^1} \otimes \cdots \otimes M_{a_N|x_N}^{A_{IO}^N} \cdot W\right],
\end{equation}
which can be interpreted as a ``generalised Born rule''. Here, $W \in A_{IO}^\N$ is a Hermitian operator called the \emph{process matrix}, which can be understood as the ``physical resource'' or ``environment'' that connects the local operations, and which satisfies the conditions
\begin{equation}
W \ge 0, \ \ W \in \L^\N, \ \ \text{and} \ \ \Tr W = \prod_{k \in \N} d_{A_O^k} \label{eq:valid_Npartite}
\end{equation}
where $\L^\N$ is a particular linear subspace of $A_{IO}^\N$~\cite{oreshkov12,araujo15} (see below).

Positive semidefiniteness $W \ge 0$ follows from the requirement that the probabilities $P(\vec a | \vec x)$ should be positive, even when the parties share an additional, possibly entangled auxiliary state $\rho$ in some extra incoming spaces $A_{I'}^1 \otimes \cdots \otimes A_{I'}^N = A_{I'}^\N$. In this case, the local operations of the parties $M_{a_k|x_k}^{A_{II'O}^k} \in A_{II'O}^k \coloneqq A_I^k \otimes A_{I'}^k \otimes A_O^k$ act on both their original incoming space and the extra incoming space, and they are composed with the ``extended'' process matrix $W \otimes \rho$ as in Eq.~\eqref{eq:born_rule_Npartite}.

The constraint $W \in \L^\N \subset A_{IO}^\N$ follows from the requirement that the probabilities $P(\vec a | \vec x)$ should be constant whenever the local operations are CPTP maps (with the third constraint $\Tr W = \prod_{k \in \N} d_{A_O^k}$ fixing this constant to $1$). The linear subspace $\L^\N$ can be characterised using the following (commutative and associative) ``trace-and-replace'' notation: 
\begin{align}
{}_{X}W \coloneqq (\Tr_X W) \otimes \frac{\id^X}{d_X}\,, \quad {}_{[1-X]}W \coloneqq W - {}_{X}W,
\end{align}
first introduced in Ref.~\cite{araujo15}. $\L^\N$ is the subspace of $A_{IO}^\N$ of all $W$ that satisfy  
\begin{align}
& \forall \ \X \subseteq \N, \X \neq \emptyset, \ {}_{ \prod_{i \in \X}[1-A_O^i]A_{IO}^{\N \backslash \X} } W = 0 \,. \label{eq:constr_LN}
\end{align}

We also recall the notion of a ``conditional (process) matrix''~\cite{Wechs2019}, defined for a given matrix $W$ and a given CP map $M^k \coloneqq M_{a_k|x_k}^{A_{IO}^k}$ applied by a party $A_k$ as
\begin{align}
W_{|M^k} \coloneqq \Tr_k \left[M^k \otimes \id^{\N \backslash k} \, \cdot \, W\right]. \label{eq:conditional_W}
\end{align}
In general, even if $W$ is a valid process matrix, $W_{|M^k}$ may not be a valid process matrix (in which case we refer to it as a ``conditional matrix'').
It has been shown that a process matrix $W$ is compatible with party $A_k$ acting first (i.e., the correlations in Eq.~\eqref{eq:born_rule_Npartite} are such that no signalling from the other parties to $A_k$ is possible) if and only if for any CP map $M^k$ the conditional matrix $W_{|M^k}$, as defined in Eq.~\eqref{eq:conditional_W}, is (up to normalisation) a valid $(N{-}1)$-partite process matrix for the parties in $\N\setminus k$~\cite{Wechs2019}. 

\subsection{Multipartite causal (non)separability}

Any process matrix $W$ satisfying Eq.~\eqref{eq:valid_Npartite} above describes a valid scenario within the process matrix framework, and may or may not be compatible with a well-defined causal order. This distinction is formalised by the concept of \emph{causal (non)separability}. 

The following general definition of multipartite causal separability was proposed in Ref.~\cite{Wechs2019}, and shown to be equivalent to that of \emph{extensible causal separability} proposed in Ref.~\cite{oreshkov16}. It has a recursive form, which formalises the intuitive notion that a process matrix is compatible with a well-defined causal order if, in any run of the experiment, one can identify a party that acts first (which party this is
can be determined probabilistically), and the conditional process matrix for the remaining parties, which depends on the action of the
first party, should again be causally separable for any
operation performed by this intial party. This accounts for the fact that, in the general multipartite case, the causal order can be both probabilistic and dynamical. Furthermore, this recursive property is required to hold even in the case of extensions with arbitrary auxiliary states, as the process matrix framework explicitly allows for such extensions, and their availability should not lead to the ``activation'' of causal nonseparability~\cite{oreshkov16,Wechs2019}.

\begin{definition}[Multipartite causal separability] \label{def:our_def-CS}
Any single-partite process matrix is causally separable. 
For $N \ge 2$, an $N$-partite process matrix $W$ is said to be \emph{causally separable} if and only if, for any extension $A_{I'}^\N$ of the parties' incoming systems and any auxiliary quantum state $\rho \in A_{I'}^\N$, $W \otimes \rho$ can be decomposed as 
\begin{equation}
 W \otimes \rho = \sum_{k \in \N} q_k \, W_{(k)}^\rho , \label{eq:our_def-CS}
\end{equation}
with $q_k\ge 0$, $\sum_k q_k = 1$, and where for each $k$, $W_{(k)}^\rho \in A_{II'O}^\N$ is a process matrix compatible with party $A_k$ acting first, and is such that for any CP map $M^k \in A_{II'O}^k$ applied by party $A_k$, the conditional $(N{-}1)$-partite process matrix%
\footnote{Note that compared to Eq.~\eqref{eq:conditional_W}, we take here $A_{II'O}^k \coloneqq A_{IO}^k \otimes A_{I'}^k$, $M^k \coloneqq M_{a_k|x_k}^{A_{II'O}^k}$, $\Tr_k \coloneqq \Tr_{A_{II'O}^k}$ and $\id^{\N \backslash k} \coloneqq \bigotimes_{j \in \N \backslash k} \id^{A_{II'O}^j}$ in the definition of the conditional matrix.}
$(W_{(k)}^\rho)_{|M^k} \coloneqq \Tr_k [M^k \otimes \id^{\N \backslash k} \, \cdot \, W_{(k)}^\rho]$ is itself causally separable.
\end{definition}

We furthermore recall the following two propositions, proven in Ref.~\cite{Wechs2019} as Propositions~B1 and~B2, respectively, which we will use in the proof below.

\begin{proposition} \label{prop:factor_rho}
Without loss of generality, each $W_{(k)}^\rho$ in Definition~\eqref{def:our_def-CS} can be taken to be of the form $W_{(k)} \otimes \rho$. Eq.~\eqref{eq:our_def-CS} then implies the direct decomposition $W = \sum_{k \in \N} q_k \, W_{(k)}$, with each $W_{(k)} \in A_{IO}^\N$ being a process matrix compatible with party $A_k$ acting first (and such that for any CP map $M^k \in A_{II'O}^k$, $(W_{(k)}\otimes\rho)_{|M^k}$ is causally separable).
\end{proposition}

\begin{proposition} \label{prop:traceout}
In a scenario where the parties' incoming spaces are decomposed as $A_I^\N \otimes A_{I'}^\N$ (possibly with some trivial spaces $A_I^k$ or $A_{I'}^k$), if a process matrix $W \in A_{II'O}^\N$ is causally separable, then so is $\Tr_{I'} W \in A_{IO}^\N$ (with $\Tr_{I'} \coloneqq \Tr_{A_{I'}^\N}$).
\end{proposition}

\section{Characterisation of multipartite causal separability}

In Proposition~5 of Ref.~\cite{Wechs2019}, a sufficient condition for multipartite causal separability was provided. We prove in the following that this condition is also necessary, i.e., we establish the following theorem.

\begin{theorem}[Characterisation of multipartite causal separability]
\label{thm:charact}
An $N$-partite process matrix $W\in A_{IO}^\Nset$ is causally separable if and only if it can be decomposed as a sum of $N!$ positive semidefinite operators $W_{(k_1,\ldots,k_N)}\ge0$ in the form
\begin{align}
 W = \sum_{(k_1,\ldots,k_N)} W_{(k_1,\ldots,k_N)}, \label{eq:causal_decomp}
\end{align}
such that for any ordered subset of parties $(k_1, \ldots, k_n)$ of 
$\Nset$ (with $1 \le n \le N$),
the partial sum
\begin{align}
 W_{(k_1,\ldots,k_n)} := \sum_{(k_{n+1},\ldots,k_N)} W_{(k_1,\ldots,k_n,k_{n+1},\ldots,k_N)} \label{eq:partial_sums}
\end{align}
satisfies
\begin{align}
 {}_{[1-A_O^{k_n}]}\Tr_{A_{IO}^{\Nset\setminus\{k_1,\ldots,k_n\}}} W_{(k_1,\ldots,k_n)}=0. \label{eq:prefix_cstrs}
\end{align}
\end{theorem}

Note that our notation here differs slightly from that of Ref.~\cite{Wechs2019}, in that we label the terms in the decomposition directly by the ordered subsets $(k_1, \ldots, k_n)$, and we merely trace out and do not replace by normalised identity operators the systems $A_{IO}^{\Nset\setminus\{k_1,\ldots,k_n\}}$ in Eq.~\eqref{eq:prefix_cstrs}. Here, as throughout this paper, the notation $(k_1, \ldots, k_n)$ implicitly assumes that all the $k_i$ are distinct; thus, the sum that appears in Eq.~\eqref{eq:causal_decomp} is over all permutations of $\Nset$, while the sum in Eq.~\eqref{eq:partial_sums} is over all of those which start with $(k_1, \ldots, k_n)$.

The ``if'' direction of Theorem~\ref{thm:charact} was proven in Ref.~\cite{Wechs2019}; here we will prove the ``only if'' part by induction. We note already at this point that the above characterisation trivially holds for $N=1$.

\subsection{Coherent teleportation technique}

The characterisation of Theorem~\ref{thm:charact} was shown to hold%
\footnote{Note that, for $N=2$, it can be shown rather trivially to hold, and corresponds directly to the bipartite characterisation of causal separability given already in Ref.~\cite{oreshkov12}.} 
for $N=3$ in Ref.~\cite{Wechs2019} using a ``teleportation technique'', wherein each pair of parties shares a maximally entangled state as part of their auxiliary systems, and the party $A_k$ acting first effectively teleports their systems $A^k_{IO}$ to a single other party (cf.\ Lemma~B1 in Ref.~\cite{Wechs2019}).
However, this technique was not sufficient to prove the characterisation beyond $N=3$, and indeed this approach was behind the necessary condition proved in Ref.~\cite{Wechs2019}, wherein separate teleportations to each other party were considered.
Linking these teleportations, obtained through different CP maps performed by the first party $A_k$, was the principal barrier to generalising this approach.

Here, we overcome this by developing a more subtle ``coherent teleportation'' technique. This consists in choosing a specific auxiliary input state, as well as a specific CP map applied by one party, such that the resulting conditional matrix is related to the original process matrix by an isometry. This isometry effectively ``teleports'' the system of the party under consideration to all other parties in a coherent manner, while attaching a ``flag'' system that indicates to which party the systems are teleported.

\medskip
\paragraph{Auxiliary state.}
For each $k_1\in\Nset$, we provide all parties $k$ with an auxiliary input system $S_{[k_1]}^k \simeq A_{IO}^{k_1}$ -- used for the teleportation of the $S$ystems $A_{IO}^{k_1}$ -- and another auxiliary input system $F_{[k_1]}^k$ of dimension $N-1$ ($\cH^{F_{[k_1]}^k} = \Span\{\ket{\ell}:\ell\in\Nset\setminus k_1\}$) -- encoding a ``$F$lag'' indicating to whom the systems will be teleported.

We then define
\begin{widetext}
\begin{align}
  \ket{\Omega_{[k_1]}} & := \frac{1}{\sqrt{N-1}} \sum_{k_2\in\Nset\setminus k_1} \!\!\ket{\Phi^+}^{S_{[k_1]}^{k_1}/S_{[k_1]}^{k_2}} \otimes \!\bigotimes_{\ell\in\Nset\setminus \{k_1,k_2\}}\!\! \ket{0}^{S_{[k_1]}^{\ell}} \otimes \bigotimes_{k\in\Nset} \ket{k_2}^{F_{[k_1]}^{k}} \label{eq:def_ket_Omega_k1}
  \quad \in \cH^{S_{[k_1]}^{\Nset}F_{[k_1]}^{\Nset}} := \bigotimes_{k\in\Nset} \cH^{S_{[k_1]}^{k}F_{[k_1]}^{k}}
\end{align}
(where $\ket{0}^{S_{[k_1]}^{\ell}}$ is any fixed state in $\cH^{S_{[k_1]}^{\ell}}$) and
\begin{align}
  \Omega & := \bigotimes_{k_1\in\Nset} \ketbra{\Omega_{[k_1]}}{\Omega_{[k_1]}} \quad \in \bigotimes_{k_1\in\Nset} S_{[k_1]}^{\Nset}F_{[k_1]}^{\Nset} = \bigotimes_{k_1,k\in\Nset} S_{[k_1]}^{k}F_{[k_1]}^{k}. \label{eq:def_Omega}
\end{align}
\end{widetext}

Note that the auxiliary system given to each party $k$ is then $A^k_{I'} := \bigotimes_{k_1\in\Nset} S_{[k_1]}^{k}F_{[k_1]}^{k}$.

\medskip
\paragraph{CP map.}

For a given party $k_1$, we define
\begin{align}
  \ket{\mu^{k_1}} & := \ket{\Phi^+}^{A_{IO}^{k_1}/S_{[k_1]}^{k_1}} \otimes \!\sum_{k_2\in\Nset\setminus k_1} \!\!\!\ket{k_2}^{F_{[k_1]}^{k_1}} \ \in \cH^{A_{IO}^{k_1}S_{[k_1]}^{k_1}F_{[k_1]}^{k_1}} \label{eq:def_ket_mu_k1}
\end{align}
and
\begin{align}
  M^{k_1} & := c_{k_1} \, \ketbra{\mu^{k_1}}{\mu^{k_1}} \otimes \bigotimes_{k'\in\Nset\setminus k_1} \id^{S_{[k']}^{k_1}F_{[k']}^{k_1}} \quad \in A_{II'O}^{k_1} \label{eq:def_Mk1}
\end{align}
with $c_{k_1} > 0$ such that $M^{k_1}$ defines the Choi matrix of a trace-nonincreasing CP map.%
\footnote{Such a value $c_{k_1} > 0$ always exists in finite dimensions (namely, one can take $c_{k_1} = d_{A_O^k}/(N-1)$). However, its precise value is irrelevant for our purposes.}

\begin{lemma}[Coherent teleportation technique]
\label{lem:coherent_teleportation}
Consider a process matrix $W \in A_{IO}^{\Nset}$ (with $N\ge 2$) to which one attaches the state $\Omega$ from Eq.~\eqref{eq:def_Omega}, and on which party $k_1$ applies the CP map $M^{k_1}$ from Eq.~\eqref{eq:def_Mk1} above. The resulting conditional matrix for the other $N-1$ parties is then
\begin{align}
  (W\otimes \Omega)_{|M^{k_1}} = \frac{c_{k_1}}{(d_{A_{IO}^k})^2} \, \widetilde J_{k_1} W \widetilde J_{k_1}^\dagger \otimes \Omega_{\setminus k_1} \label{eq:coherent_teleportation}
\end{align}
with
\begin{align}
  \widetilde J_{k_1} & := J_{k_1} \otimes \id^{A_{IO}^{\Nset\setminus k_1}}, \notag \\
  J_{k_1} & := \frac{1}{\sqrt{N-1}} \sum_{k_2\in\Nset\setminus k_1} J_{k_1,k_2}, \notag \\
  J_{k_1,k_2} & := \id^{A_{IO}^{k_1}\to S_{[k_1]}^{k_2}} \otimes \!\bigotimes_{\ell\in\Nset\setminus \{k_1,k_2\}} \!\!\!\ket{0}^{S_{[k_1]}^{\ell}} \otimes \!\bigotimes_{\ell'\in\Nset\setminus k_1} \!\!\!\ket{k_2}^{F_{[k_1]}^{\ell'}}\!,
\end{align}
and $\Omega_{\setminus k_1} := \bigotimes_{k'\in\Nset\setminus k_1} \Tr_{S_{[k']}^{k_1}F_{[k']}^{k_1}} \ketbra{\Omega_{[k']}}{\Omega_{[k']}}$.
\end{lemma}

The proof of this lemma is obtained through a series of algebraic manipulations, which we relegate to the \hyperref[app:proof_coherent_telep_lemma]{Appendix}. Notice already that $\widetilde J_{k_1}$, $J_{k_1}$ and $J_{k_1,k_2}$ are all isometries.

Using Definition~\ref{def:our_def-CS} (multipartite causal separability) together with Propositions~\ref{prop:factor_rho} and~\ref{prop:traceout} recalled above, the following then immediately follows:

\begin{corollary}[Necessary condition from coherent teleportation]
\label{cor:coherent_teleportation_NC}
A causally separable process matrix $W \in A_{IO}^{\Nset}$ (with $N\ge 2$) necessarily admits a decomposition 
\begin{align}
  W = \sum_{k_1\in\Nset} W_{[k_1]},
\end{align}
where each term $W_{[k_1]} \ge 0$ is a (subnormalised) process matrix compatible with party $k_1$ acting first, and such that $\overline{W}_{[k_1]} := \widetilde J_{k_1} W_{[k_1]} \widetilde J_{k_1}^\dagger \in A_{II'O}^{\Nset\setminus k_1}$ is a causally separable process matrix for the parties in $\Nset\setminus k_1$ (with now input systems $A_{II'}^{\ell} := A_I^{\ell} S_{[k_1]}^{\ell} F_{[k_1]}^{\ell}$ for all $\ell\in\Nset\setminus k_1$).
\end{corollary}

\subsection{Recursive proof of the ``only if'' part in Theorem~\ref{thm:charact}}

Having introduced the coherent teleportation technique and its consequence in the form of Corollary~\ref{cor:coherent_teleportation_NC}, we can now present the core of the proof for the ``only if'' part in Theorem~\ref{thm:charact}.

\begin{proof}

Consider a causally separable process matrix $W \in A_{IO}^{\Nset}$, with $N\ge 2$.

Assume that the characterisation of Theorem~\ref{thm:charact} holds true for ($N-1$)-partite causally separable process matrices. Each term $W_{[k_1]}$ that appears in the decomposition of Corollary~\ref{cor:coherent_teleportation_NC} is then such that $\overline{W}_{[k_1]} := \widetilde J_{k_1} W_{[k_1]} \widetilde J_{k_1}^\dagger$ admits a decomposition as a sum of $(N-1)!$ positive semidefinite operators $\overline{W}_{[k_1],(k_2,\ldots,k_N)}\ge0$ in the form
\begin{align}
 \overline{W}_{[k_1]} = \sum_{(k_2,\ldots,k_N)} \overline{W}_{[k_1],(k_2,\ldots,k_N)},
\end{align}
such that for any ordered subset of parties $(k_2, \ldots, k_n)$ of 
$\Nset\setminus k_1$ (with $2 \le n \le N$), the partial sum
\begin{align}
 \overline{W}_{[k_1],(k_2,\ldots,k_n)} := \sum_{(k_{n+1},\ldots,k_N)} \overline{W}_{[k_1],(k_2,\ldots,k_N)},
\end{align}
satisfies
\begin{align}
 {}_{[1-A_O^{k_n}]}\Tr_{A_{II'O}^{\Nset\setminus\{k_1,\ldots,k_n\}}} \overline{W}_{[k_1],(k_2,\ldots,k_n)}=0. \label{eq:prefix_cstrs_barW}
\end{align}

Define then
\begin{align}
 W_{(k_1,k_2,\ldots,k_N)} := \widetilde J_{k_1}^\dagger \overline{W}_{[k_1],(k_2,\ldots,k_N)} \widetilde J_{k_1} \quad \in A_{IO}^{\Nset},
\end{align}
which are such that $W_{(k_1,k_2,\ldots,k_N)} \ge 0$ and
\begin{align}
 \sum_{(k_1,k_2,\ldots,k_N)} \!\!\!\! & W_{(k_1,k_2,\ldots,k_N)} = \sum_{(k_1,k_2,\ldots,k_N)} \!\!\!\!\widetilde J_{k_1}^\dagger \overline{W}_{[k_1],(k_2,\ldots,k_N)} \widetilde J_{k_1} \notag \\
 & = \ \sum_{k_1} \widetilde J_{k_1}^\dagger \Big(\sum_{(k_2,\ldots,k_N)} \overline{W}_{[k_1],(k_2,\ldots,k_N)}\Big) \widetilde J_{k_1} \notag \\
 & = \ \sum_{k_1} \widetilde J_{k_1}^\dagger \overline{W}_{[k_1]} \widetilde J_{k_1} \ = \ \sum_{k_1} W_{[k_1]} \ = \ W
\end{align}
(where we used the fact that the $\widetilde J_{k_1}$ are isometries).

Defining the partial sums $W_{(k_1,\ldots,k_n)}$ as in Eq.~\eqref{eq:partial_sums}, first note that $W_{(k_1)} = \sum_{(k_2,\ldots,k_N)}W_{(k_1,k_2,\ldots,k_N)} = \widetilde J_{k_1}^\dagger \big(\sum_{(k_2,\ldots,k_N)} \overline{W}_{[k_1],(k_2,\ldots,k_N)}\big) \widetilde J_{k_1} = \widetilde J_{k_1}^\dagger \overline{W}_{[k_1]} \widetilde J_{k_1} = W_{[k_1]}$, and that the fact that $W_{[k_1]}$ is a (subnormalised) process matrix compatible with party $k_1$ acting first implies in particular that ${}_{[1-A_O^{k_1}]}\Tr_{A_{IO}^{\Nset\setminus k_1}} W_{(k_1)}=0$ (see e.g. Eq.~(A13) of~\cite{Wechs2019}) -- so that the constraint of Eq.~\eqref{eq:prefix_cstrs} is satisfied for $n=1$.

\bigskip

Before verifying the same constraint for $n>1$, first note that the decomposition of $\overline{W}_{[k_1]} = \sum_{(k_2,\ldots,k_N)} \overline{W}_{[k_1],(k_2,\ldots,k_N)}$ implies that
\begin{align}
 0 \le \overline{W}_{[k_1],(k_2,\ldots,k_N)} \le \overline{W}_{[k_1]},
\end{align}
which further implies that $\ker \overline{W}_{[k_1]} \subseteq \ker \overline{W}_{[k_1],(k_2,\ldots,k_N)}$ and therefore, together with $\overline{W}_{[k_1]} = \widetilde J_{k_1} W_{[k_1]} \widetilde J_{k_1}^\dagger$, that
\begin{align}
 \supp \overline{W}_{[k_1],(k_2,\ldots,k_N)} \subseteq \supp \overline{W}_{[k_1]} \subseteq \range \widetilde J_{k_1}.
\end{align}
It then follows that $(\widetilde J_{k_1} \widetilde J_{k_1}^\dagger) \overline{W}_{[k_1],(k_2,\ldots,k_N)} (\widetilde J_{k_1} \widetilde J_{k_1}^\dagger) = \overline{W}_{[k_1],(k_2,\ldots,k_N)}$, i.e.\
\begin{align}
 \overline{W}_{[k_1],(k_2,\ldots,k_N)} = \widetilde J_{k_1} W_{(k_1,k_2,\ldots,k_N)} \widetilde J_{k_1}^\dagger
\end{align}
and taking the partial sums,
\begin{align}
 \overline{W}_{[k_1],(k_2,\ldots,k_n)} = \widetilde J_{k_1} W_{(k_1,k_2,\ldots,k_n)} \widetilde J_{k_1}^\dagger
\end{align}
for all $n=2,\ldots,N$.

The constraint of Eq.~\eqref{eq:prefix_cstrs_barW} can then be written
\begin{align}
 {}_{[1-A_O^{k_n}]}\Tr_{A_{II'O}^{\Nset\setminus\{k_1,\ldots,k_n\}}} \big( \widetilde J_{k_1} W_{(k_1,k_2,\ldots,k_n)} \widetilde J_{k_1}^\dagger \big) =0,
\end{align}
and, multiplying by $\ketbra{k_2}{k_2}^{F_{[k_1]}^{k_2}}$ on both sides of the LHS (noting that $\ketbra{k_2}{k_2}^{F_{[k_1]}^{k_2}} \widetilde J_{k_1} = \frac{1}{\sqrt{N-1}} \widetilde J_{k_1,k_2}$, with implicit identity operators and $\widetilde J_{k_1,k_2} := J_{k_1,k_2} \otimes \id^{A_{IO}^{\Nset\setminus k_1}}$), one obtains
\begin{align}
 {}_{[1-A_O^{k_n}]}\Tr_{A_{II'O}^{\Nset\setminus\{k_1,\ldots,k_n\}}} \big( \widetilde J_{k_1,k_2} W_{(k_1,k_2,\ldots,k_n)} \widetilde J_{k_1,k_2}^\dagger \big) =0.
\end{align}
Direct calculation then gives
\begin{widetext}
\begin{align}
 & {}_{[1-A_O^{k_n}]}\Tr_{A_{II'O}^{\Nset\setminus\{k_1,\ldots,k_n\}}} \big( \widetilde J_{k_1,k_2} W_{(k_1,k_2,\ldots,k_n)} \widetilde J_{k_1,k_2}^\dagger \big) \notag \\
 & \quad = \Tr_{A_{I'}^{\Nset\setminus\{k_1,\ldots,k_n\}}} \Big( J_{k_1,k_2} \otimes \id^{A_{IO}^{\{k_2,\ldots,k_n\}}} \big( {}_{[1-A_O^{k_n}]}\Tr_{A_{IO}^{\Nset\setminus\{k_1,\ldots,k_n\}}} W_{(k_1,k_2,\ldots,k_n)} \big) J_{k_1,k_2}^\dagger \otimes \id^{A_{IO}^{\{k_2,\ldots,k_n\}}} \Big) \notag \\
 & \quad = \Big( \id^{A_{IO}^{k_1}\to S_{[k_1]}^{k_2}} \otimes \id^{A_{IO}^{\{k_2,\ldots,k_n\}}} \big( {}_{[1-A_O^{k_n}]}\Tr_{A_{IO}^{\Nset\setminus\{k_1,\ldots,k_n\}}} W_{(k_1,k_2,\ldots,k_n)} \big) \id^{S_{[k_1]}^{k_2}\to A_{IO}^{k_1}} \otimes \id^{A_{IO}^{\{k_2,\ldots,k_n\}}} \Big) \notag \\
 & \hspace{30mm} \otimes \bigotimes_{\ell\in\{k_3,\ldots,k_n\}} \ketbra{0}{0}^{S_{[k_1]}^{\ell}} \otimes \bigotimes_{\ell'\in\{k_2,\ldots,k_n\}} \ketbra{k_2}{k_2}^{F_{[k_1]}^{\ell'}} = 0,
\end{align}
\end{widetext}
which, after tracing out the auxiliary systems $S_{[k_1]}^{\ell}$ and $F_{[k_1]}^{\ell'}$ on the last line, implies that
\begin{align}
 {}_{[1-A_O^{k_n}]}\Tr_{A_{IO}^{\Nset\setminus\{k_1,\ldots,k_n\}}} W_{(k_1,k_2,\ldots,k_n)} = 0,
\end{align}
i.e.\ that the constraint of Eq.~\eqref{eq:prefix_cstrs} is indeed also satisfied for $n>1$ -- which, by recursion, concludes the proof.
\end{proof}

\section{Discussion}

We resolved the main open question from Ref.~\cite{Wechs2019}, proving that the sufficient condition for causal separability provided therein is also a necessary condition, and hence provides a complete characterisation of causally separable process matrices.
It was shown in Ref.~\cite{wechs18} that the class of generalised quantum circuits with classical control of causal order, or QC-CCs, correspond precisely to the sufficient condition of Ref.~\cite{Wechs2019}, so our result shows that all causally separable processes have possible realisations as QC-CCs -- as also originally conjectured in Ref.~\cite{oreshkov16}.
This closes the question of the physical realisability of causally separable processes, since all QC-CCs can be given concrete physical realisations.
Moreover, the characterisation of Theorem~\ref{thm:charact} can be verified with semidefinite programming, so our results confirm that one can efficiently optimise over causally separable processes and efficiently find causal witnesses for any causally nonseparable process~\cite{Wechs2019}.

The key insight in the proof of Theorem~\ref{thm:charact} was the introduction of the ``coherent teleportation technique'' of Lemma~\ref{lem:coherent_teleportation}.
We are not aware of previous use of such a technique, which could be useful in other contexts, and it would be interesting to explore this direction further.
Variations of this technique could also be considered. For instance, the information about whose party $k_2$ the systems of $A_{k_1}$ is teleported to, in each branch (encoded in $\bigotimes_{k\in\Nset} \ket{k_2}^{F_{[k_1]}^{k}}$), may not need to be accessible to all other parties: it should be enough that only $k_1$ knows $k_2$ ($k_1$ should still receive $\ket{k_2}^{F_{[k_1]}^{k_1}}$), and that each other party $A_\ell$ only knows if they are the receiver (if they are $k_2$) or not (i.e.\ for them, it would suffice that $F_{[k_1]}^{\ell}$ is a qubit encoding this binary information). We leave the exploration of the full potentiality of the coherent teleportation approach for future research.

\bigskip

\paragraph*{Note added.}
While completing this work, we became aware of Ref.~\cite{wei2026}, which proves the same result as ours, also using a ``generalized teleportation construction''.

\medskip

\paragraph*{Use of AI statement/disclosure.}
AI models were used to obtain an initial proof of Theorem~\ref{thm:charact}. In particular, OpenAI's GTP-5.6-sol produced a proof that identified the ``coherent teleportation technique'', after some gentle encouragement to try harder, and helped simplify the proof somewhat. The proof it produced was however more unwieldy than the one presented here, and the authors further simplified and clarified the argument. The manuscript was entirely written and prepared by the authors, who take full responsibility for the correctness and content of this article.

\medskip

\paragraph*{Acknowledgements.}

This research was funded in part by l’Agence Nationale de la Recherche (ANR) projects ANR-15-IDEX-02 and ANR-22-CE47-0012, and the PEPR integrated project EPiQ ANR-22-PETQ-0007 as part of Plan France 2030.

For the purpose of open access, the authors have applied a CC-BY public copyright license to any Author Accepted Manuscript (AAM) version arising from this submission.

\bibliography{literature}

\appendix
\onecolumngrid

\subsection*{Appendix: Proof of Lemma~\ref{lem:coherent_teleportation}}
\label{app:proof_coherent_telep_lemma}

Here we prove the form of the resulting matrix in Lemma~\ref{lem:coherent_teleportation}, Eq.~\eqref{eq:coherent_teleportation}.
In the calculations below we freely change the order of tensor products, keeping track of their labels and of which ones get multiplied together.

\begin{proof}

From the definition of a conditional matrix, Eq.~\eqref{eq:conditional_W}, of the state $\Omega$~\eqref{eq:def_Omega}, and of the CP map $M^{k_1}$~\eqref{eq:def_Mk1}:
\begin{align}
(W\otimes \Omega)_{|M^{k_1}} & = \Tr_{k_1} \left[ \big( M^{k_1} \otimes \id^{\N \backslash k_1} \big) \, \cdot \, (W\otimes \Omega)\right] \notag \\
& = c_{k_1} \, \Tr_{k_1} \left[ \Big( \ketbra{\mu^{k_1}}{\mu^{k_1}} \otimes \bigotimes_{k'\in\Nset\setminus k_1} \id^{S_{[k']}^{k_1}F_{[k']}^{k_1}} \otimes \id^{\N \backslash k_1} \Big) \, \cdot \, \Big( W\otimes \bigotimes_{k'\in\Nset} \ketbra{\Omega_{[k']}}{\Omega_{[k']}} \Big) \right] \notag \\
& = c_{k_1} \, \Tr_{k_1} \left[ \Big( \ketbra{\mu^{k_1}}{\mu^{k_1}} \otimes \bigotimes_{k'\in\Nset\setminus k_1} \id^{S_{[k']}^{k_1}F_{[k']}^{k_1}} \otimes \id^{\N \backslash k_1} \Big) \, \cdot \, \Big( W\otimes \ketbra{\Omega_{[k_1]}}{\Omega_{[k_1]}} \otimes \bigotimes_{k'\in\Nset\setminus k_1} \ketbra{\Omega_{[k']}}{\Omega_{[k']}} \Big) \right] \notag \\
& = c_{k_1} \, \Tr_{k_1} \left[ \Big( \ketbra{\mu^{k_1}}{\mu^{k_1}} \otimes \id^{\N \backslash k_1} \Big) \, \cdot \, \Big( W\otimes \ketbra{\Omega_{[k_1]}}{\Omega_{[k_1]}} \Big) \right] \otimes \underbrace{\bigotimes_{k'\in\Nset\setminus k_1} \Tr_{S_{[k']}^{k_1}F_{[k']}^{k_1}} \ketbra{\Omega_{[k']}}{\Omega_{[k']}}}_{\Omega_{\setminus k_1}} \label{eq:proof_app_1}
\end{align}
(where $\id^{\N \backslash k_1}$ is the identity on all systems attributed to parties in $\N \backslash k_1$ -- not necessarily the same systems in the different expressions).

Introducing the computational basis $\{\ket{i}^X\}_{i}$ for the isomorphic spaces $\cH^{X^{(\prime)}} = \cH^{A_{IO}^{k_1}}, \cH^{S_{[k_1]}^{k_1}}, \cH^{S_{[k_1]}^{k_2}}$, recalling the expression for the maximally entangled state across any two of these spaces, $\ket{\Phi^+}^{XX'} \coloneqq \frac{1}{\sqrt{d_{X^{(\prime)}}}}\sum_i \ket{i}^X\otimes \ket{i}^{X'}$, and using the expressions of Eqs.~\eqref{eq:def_ket_Omega_k1} and~\eqref{eq:def_ket_mu_k1}, we can write, for some vector $\ket{w}\in\cH^{A_{IO}^\N}$:
\begin{align}
& \big( \bra{\mu^{k_1}} \otimes \id^{\N \backslash k_1} \big) \cdot \big( \ket{w}\otimes \ket{\Omega_{[k_1]}} \big) \notag \\
& = \frac{1}{\sqrt{N-1}} \Big( \!\bra{\Phi^+}^{A_{IO}^{k_1}/S_{[k_1]}^{k_1}} \otimes \!\!\sum_{k_2'\in\Nset\setminus k_1} \!\!\!\!\bra{k_2'}^{F_{[k_1]}^{k_1}} \otimes \id^{\N \backslash k_1} \Big) \cdot \Big( \ket{w} \otimes \!\sum_{k_2\in\Nset\setminus k_1} \!\!\!\ket{\Phi^+}^{S_{[k_1]}^{k_1}/S_{[k_1]}^{k_2}} \otimes \!\bigotimes_{\ell\in\Nset\setminus \{k_1,k_2\}}\!\!\! \ket{0}^{S_{[k_1]}^{\ell}} \otimes \bigotimes_{k\in\Nset} \ket{k_2}^{F_{[k_1]}^{k}} \Big) \notag \\
& = \frac{1}{d_{A_{IO}^{k_1}}\sqrt{N-1}} \sum_{k_2\in\Nset\setminus k_1} \sum_{i,i'} \Big(\! \bra{i'}^{A_{IO}^{k_1}} \otimes \bra{i'}^{S_{[k_1]}^{k_1}} \otimes \id^{\N \backslash k_1} \Big) \!\cdot\! \Big( \ket{w} \otimes \ket{i}^{S_{[k_1]}^{k_1}} \otimes \ket{i}^{S_{[k_1]}^{k_2}} \otimes \!\!\bigotimes_{\ell\in\Nset\setminus \{k_1,k_2\}}\!\!\! \ket{0}^{S_{[k_1]}^{\ell}} \otimes \!\bigotimes_{k\in\Nset\setminus k_1}\! \ket{k_2}^{F_{[k_1]}^{k}} \Big) \notag \\
& = \frac{1}{d_{A_{IO}^{k_1}}\sqrt{N-1}} \sum_{k_2\in\Nset\setminus k_1} \Big( \sum_{i} \ket{i}^{S_{[k_1]}^{k_2}}\bra{i}^{A_{IO}^{k_1}} \otimes \id^{A_{IO}^{\Nset\setminus k_1}} \Big)\ket{w} \otimes \!\bigotimes_{\ell\in\Nset\setminus \{k_1,k_2\}}\!\! \ket{0}^{S_{[k_1]}^{\ell}} \otimes \bigotimes_{k\in\Nset\setminus k_1} \ket{k_2}^{F_{[k_1]}^{k}} \notag \\
& = \frac{1}{d_{A_{IO}^{k_1}}} \Big( \underbrace{\frac{1}{\sqrt{N-1}} \sum_{k_2\in\Nset\setminus k_1} \underbrace{\id^{A_{IO}^{k_1}\to S_{[k_1]}^{k_2}} \otimes \!\bigotimes_{\ell\in\Nset\setminus \{k_1,k_2\}}\!\! \ket{0}^{S_{[k_1]}^{\ell}} \otimes \bigotimes_{k\in\Nset\setminus k_1} \ket{k_2}^{F_{[k_1]}^{k}}}_{J_{k_1,k_2}}}_{J_{k_1}} \otimes \id^{A_{IO}^{\Nset\setminus k_1}} \Big) \ket{w} = \frac{1}{d_{A_{IO}^{k_1}}} \widetilde J_{k_1} \ket{w}.
\end{align}

Considering some decomposition $W = \sum_{w,w'} \ketbra{w}{w'}$ of $W$, and injecting the result of the above computation back into Eq.~\eqref{eq:proof_app_1}, we get the expression of Eq.~\eqref{eq:coherent_teleportation}.
\end{proof}

\end{document}